\documentclass[conference]{IEEEtran}
\IEEEoverridecommandlockouts
\usepackage{cite}
\usepackage{dblfloatfix}
\usepackage{overpic}
\usepackage{amsmath,amssymb,amsfonts}
\usepackage{graphicx}
\usepackage{textcomp}
\usepackage{xcolor}
\usepackage{float}
\usepackage{amsthm}
\usepackage{graphicx}
\usepackage{epstopdf}
\usepackage{amsmath,bm}
\usepackage{amsfonts}
\usepackage{amssymb}
\usepackage{color}
\usepackage{multirow}
\usepackage{multicol}
\usepackage{soul,xcolor}
\usepackage{algorithm}
\usepackage{algpseudocode}

\theoremstyle{plain}

\newtheorem{lemma}{Lemma}

\newcommand{\vect}[1]{\mathbf{#1}}

\def\tr{\mathrm{tr}}
\def\rank{\mathrm{rank}}
\def\Htran{\mbox{\tiny $\mathrm{H}$}}
\def\Ttran{\mbox{\tiny $\mathrm{T}$}}
\def\CN{\mathcal{N}_{\mathbb{C}}} 
\def\imagunit{\mathsf{j}} 
\def\m{\rm}

\def\sinc{\mathrm{sinc}}

\begin{document}

\title{Pilot Length Optimization with RS-LS Channel Estimation for Extremely Large Aperture Arrays 
\thanks{The work by M. Alıcıo\u{g}lu was supported by ASELSAN Inc., Ankara, T\"urkiye. The work by \"O. T. Demir was supported by 2232-B International Fellowship for Early Stage Researchers Programme funded by the Scientific and Technological Research Council of T\"urkiye. The work by E.~Bj\"ornson was supported by the  FFL18-0277 grant from the Swedish Foundation for Strategic Research.}}

\author{\IEEEauthorblockN{Mert Alıcıo\u{g}lu$^*$, \"Ozlem Tu\u{g}fe Demir$^*$, Emil Bj{\"o}rnson$^{\dagger}$}
\IEEEauthorblockA{{$^*$Department of Electrical-Electronics Engineering, TOBB University of Economics and Technology, Ankara, T\"urkiye
		} \\ {$^\dagger$Department of Computer Science, KTH Royal Institute of Technology, Kista, Sweden  
		} \\
		{Email: malicioglu@etu.edu.tr, ozlemtugfedemir@etu.edu.tr, emilbjo@kth.se}
}

}

\maketitle
\begin{abstract}
Extremely large aperture arrays can enable unprecedented spatial multiplexing in beyond 5G systems due to their extremely narrow beamfocusing capabilities. However, acquiring the spatial correlation matrix to enable efficient channel estimation is a complex task due to the vast number of antenna dimensions. Recently, a new estimation method called the ``reduced-subspace least squares (RS-LS) estimator'' has been proposed for densely packed arrays. This method relies solely on the geometry of the array to limit the estimation resources. In this paper, we address a gap in the existing literature by deriving the average spectral efficiency for a certain distribution of user equipments (UEs) and a lower bound on it when using the RS-LS estimator. This bound is determined by the channel gain and the statistics of the normalized spatial correlation matrices of potential UEs but, importantly, does not require knowledge of a specific UE's spatial correlation matrix. We establish that there exists a pilot length that maximizes this expression. Additionally, we derive an approximate expression for the optimal pilot length under low signal-to-noise ratio (SNR) conditions. Simulation results validate the tightness of the derived lower bound and the effectiveness of using the optimized pilot length.   
\end{abstract}

\vspace{-2mm}
\section{Introduction}

To facilitate effective beamforming and spatial multiplexing for user equipments (UEs), 5G base stations (BSs) are equipped with a multitude of antennas. Massive MIMO (multiple-input multiple-output) is the term used for this technology, which is characterized by having many more antennas than UEs to enhance spectral efficiency through spatial multiplexing~\cite{massivemimobook}. The advantages of massive MIMO are most pronounced when the antenna array dimensions are extremely large \cite{gao2013massive,martinez2014towards}.  However, the aperture size is typically limited by practical constraints, necessitating densely packed arrays when a very high number of antennas is desired. This involves the use of hundreds of antennas with antenna spacing possibly less than half the wavelength, referred to as \emph{extremely large aperture arrays (ELAA)} \cite{Bjornson2019d}, \emph{holographic MIMO} \cite{huang2020holographic, pizzo2020spatially}, or \emph{large intelligent surfaces} \cite{hu2018beyond}. By increasing the number of antennas in a given aperture, one can approach the asymptotic spatial degrees-of-freedom (DoF), reduce interference, and achieve massive array gains in all directions \cite{Bjornson2019d, hu2018beyond, pizzo2020asilomar, ramezani2023massive}.

One of the challenges associated with ELAA is channel estimation, which requires vast signal resources unless the channel structure and, particularly, sparsity are exploited. Such sparsity is both caused by clustered scattering and spatial oversampling in the array.
The spatial correlation matrix captures both effects and is used by the minimum mean squared error (MMSE) estimator, but acquiring this high-dimensional matrix is difficult in practice. If the scattering is extremely sparse, parametric channel structures can be leveraged \cite{an2023tutorial, ghermezcheshmeh2023parametric} and possibly compressed sensing, but measured channels are not that simple \cite{Gao2015b}.
If the channel structure is entirely unknown, the least squares (LS) estimator is conventionally used, but it performs much worse. Recently, the reduced-subspace least squares (RS-LS) estimator was introduced in \cite{demir2022channel}. It utilizes the spatial correlation created by the array geometry to estimate the channel in a lower-dimensional subspace where it resides, without the need for any UE-specific spatial correlation knowledge. RS-LS outperforms the LS estimator by depressing noise and requiring fewer pilot resources. Subsequently, its performance has been explored in the context of reconfigurable intelligent surface-aided communications \cite{long2023channel}. 

While the performance of the RS-LS channel estimator has been studied in terms of estimation errors, the analysis of the resulting spectral efficiency (SE) has been notably absent in the literature. In this paper, we address this gap by deriving the SE of a UE whose channel is estimated using the RS-LS estimator for multi-length pilot transmissions, with a particular emphasis on low signal-to-noise ratio (SNR) scenarios. We subsequently develop a novel expression for the average SE and a closed-form lower bound dependent only on the statistics of the normalized spatial correlation matrices. This allows us to determine the optimal pilot length for a given UE distribution and channel gain. Our analytical findings demonstrate that there exists a unique pilot length that maximizes the average SE. We derive a closed-form expression for optimal pilot length at low SNRs. The simulation results validate the tightness of the proposed expressions and demonstrate the near-optimal performance of the closed-form pilot length. These analytical results provide new insights into how channel estimation can be conducted efficiently without requiring complete spatial correlation knowledge in systems with ELAAs.

\vspace{-1mm}

\section{System and Channel Modeling}
\vspace{-1mm}

We examine the uplink operation of a single-antenna UE to a BS equipped with an ELAA. The array is constructed as a uniform planar array (UPA) consisting of $M$ antennas.  We specify that $M_{\rm H}$ antennas are arranged in each row, and $M_{\rm V}$ antennas are placed in each column, thereby resulting in a total of $M=M_{\rm H}M_{\rm V}$ antennas. The horizontal and vertical separation between these antennas is denoted as $\Delta$. 
Our main focus is on scenarios where there are hundreds of antennas with an antenna spacing that is less than half of the wavelength $\lambda$, as needed to approach the asymptotic DoF limits.

Using row-by-row numbering by $m\in[1,M]$, as exemplified in \cite{demir2022channel}, the position of the $m$th antenna relative to the origin 
is expressed as $\vect{u}_m = [ 0, \, \,\, i(m) \Delta,  \,\,\, j(m) \Delta]^{\Ttran}$. Here, $i(m) =\mathrm{mod}(m-1,M_{\m H})$ and $j(m) =\left\lfloor(m-1)/M_{\m H}\right\rfloor$
are the horizontal and vertical indices of antenna element $m$, respectively, while $\mathrm{mod}(\cdot,\cdot)$ denotes the modulus operation and $\lfloor \cdot \rfloor$ truncates the argument. If a plane wave is incident on the UPA from the azimuth angle $\varphi$ and elevation angle $\theta$, the array response vector is described as \cite[Sec.~7.3]{massivemimobook}
\begin{equation}\label{eq:array-response}
\vect{a}(\varphi,\theta) = \left[e^{\imagunit\vect{k}^{\Ttran}(\varphi,\theta)\vect{u}_1},\dots,e^{\imagunit\vect{k}^{\Ttran}(\varphi,\theta)\vect{u}_M}\right]^{\Ttran},
\end{equation}
where the wave vector is $\vect{k}(\varphi, \theta) = \frac{2\pi}{\lambda}\left[\cos(\theta) \cos(\varphi), \,\,\, \cos(\theta) \sin(\varphi), \,\,\, \sin(\theta)\right]^{\Ttran}$.

Let us denote the channel of an arbitrary single-antenna UE to the BS by $\vect{h}\in \mathbb{C}^M$. This channel typically comprises a superposition of multipath components that can be expanded as a continuum of plane waves \cite{Sayeed2002a}, which holds even if some components give spherical waves. We consider the conventional block fading model, where the channel $\vect{h}$ is constant within a time-frequency block and takes independent realization across blocks from a stationary stochastic distribution.
In accordance with \cite{Sayeed2002a}, we consider a correlated Rayleigh fading channel given as
\begin{equation} \label{eq:corr-Rayleigh}
\vect{h} \sim \CN(\vect{0},\beta\vect{R}),
\end{equation}
where $\beta$ denotes the average channel gain (i.e., capturing pathloss and shadowing) and $\vect{R}\in \mathbb{C}^{M \times M}$ is the normalized spatial correlation matrix so that $\tr(\vect{R})=M$. 

 In this paper, we will optimize the pilot length based on the average SE obtained among all potential UE locations, which can be easily obtained in practice by collecting statistics at the BS throughout a long period. Such an idea of average SE was previously considered in \cite{hossain2017energy,lopez2021energy,Zappone2023tradeoff}. There is a certain distribution of UEs and they have a distribution of $\vect{R}$ matrices. Hence, an arbitrary UE channel is described using a normalized spatial correlation matrix, whose realization is unknown at the BS, but the statistics of the spatial correlation matrices is known.  

Let us focus on an arbitrary UE with the channel gain $\beta$, which is known at the BS, and unknown normalized spatial correlation matrix $\vect{R}$. Let $f(\varphi,\theta)$ represent the normalized \emph{spatial scattering function} \cite{Sayeed2002a}.  This function describes the angular multipath distribution and the directivity gain of the antennas, and it is normalized so that $\iint f(\varphi,\theta) d\theta d\varphi  = 1$ since it behaves as a joint probability density function with respect to azimuth and elevation angles. The normalized spatial correlation matrix $\vect{R}$ for the considered UE depends entirely on $f(\varphi,\theta)$ and the array response vectors, because it can be expressed as 
\begin{equation} \label{eq:spatial-correlation}
\vect{R} = \mathbb{E}\{ \vect{h} \vect{h}^{\Htran} \} =   \iint_{-\pi/2}^{\pi/2} f(\varphi,\theta) \vect{a}(\varphi,\theta) \vect{a}^{\Htran}(\varphi,\theta) d\theta d\varphi ,
\end{equation}
where $\tr(\vect{R}) = M$. Note that the waves only arrive from directions in front of the array; that is, $\varphi \in [-\frac{\pi}{2},\frac{\pi}{2}]$. From \cite[Lem.~1]{demir2022channel}, the $(m,l)$th entry of the spatial correlation matrix is given as
\begin{align}  \label{eq:spatial-correlation2}
\left[\vect{R}\right]_{m,l} = &
 \iint_{-\pi/2}^{\pi/2} f(\varphi,\theta)  e^{\imagunit2\pi \left( d_{{\rm H}}^{ml}\sin(\varphi)\cos(\theta) + d_{{\rm V}}^{ml}\sin(\theta) \right)} d\theta d\varphi ,
\end{align}
where the horizontal and vertical distances between antenna $m$ and $l$ in numbers of the wavelength are given by
\begin{align}
d_{{\rm H}}^{ml} = \frac{ \left( i(m) - i(l) \right) \Delta}{\lambda}, \quad d_{{\rm V}}^{ml} = \frac{ \left( j(m) - j(l) \right) \Delta}{\lambda}.
\end{align}

The double integral described in \eqref{eq:spatial-correlation2} can be computed numerically for various spatial scattering functions, including those describing spherical waves. However, for some functions, closed-form expressions are also possible. One such case is in an ``isotropic scattering environment,'' where multipath components exhibit equal strength in all directions, and the antennas are isotropic, denoted by $f(\varphi,\theta) = \cos(\theta) / (2\pi)$. Here, the cosine term arises from the use of spherical coordinates. We denote the resulting normalized correlation matrix as $\vect{R}_{\rm iso}$, and the $(m,l)$th entry is~\cite{demir2022channel}:
\begin{equation}\label{R-iso}
\left[\vect{R}_{\rm iso} \right]_{m,l} = \sinc \left( 2 \sqrt{\left(d_{\rm H}^{ml}\right)^2+\left(d_{\rm V}^{ml}\right)^2}\right),
\end{equation}
where $\sinc(x) = \sin(\pi x)/ (\pi x)$ represents the sinc function. The expression in \eqref{R-iso} reveals that when two antennas are spaced apart by an integer multiple of $\lambda/2$, they will experience uncorrelated fading. However, we stress that this condition cannot be satisfied for all pairs of antennas in a UPA \cite{demir2022channel}. Hence, such an array will always exhibit spatially correlated fading. This correlation is not caused by correlation in the scattering environment but by the array geometry.

\section{Channel Estimation} 

The BS must estimate $\vect{h}$ in each time-frequency coherence block to perform coherent combining using all the $M$ antennas. 
We let $\tau_c$ denote the number of channel uses per coherence block, and let the UE send a predefined pilot sequence that spans $\tau_p<\tau_c$ channel uses. From \cite[Sec.~3]{massivemimobook}, the received signal at the BS during $\tau_{\rm p}$ pilot channel uses  is\footnote{When multiple (up to $\tau_p$) UEs send orthogonal pilots, the channel estimation is done separately by despreading the received signal at the BS. } 
\begin{align} \label{eq:y}
    \vect{Y}^{\rm pilot} = \sqrt{\rho}\vect{h}\boldsymbol{\phi}^{\Ttran} + \vect{N}^{\rm pilot},
\end{align}
where $\rho>0$ is the uplink SNR and $\vect{N}^{\rm pilot}\in\mathbb{C}^{M\times \tau_{\rm p}}$ includes the independent and identically distributed $\CN(0,1)$ entries. The pilot vector $\boldsymbol{\phi}\in \mathbb{C}^{\tau_{\rm p}}$ satisfies $\boldsymbol{\phi}^{\Ttran}\boldsymbol{\phi}^*=\|\boldsymbol{\phi}\|^2=\tau_{\rm p}$. The BS multiplies the received signal in \eqref{eq:y} by the unit-norm vector $\boldsymbol{\phi}^*/\|\boldsymbol{\phi}\|$ and obtains the sufficient statistics for estimating $\vect{h}$:
\begin{align}
 \vect{y}^{\rm pilot}=\vect{Y}^{\rm pilot} \frac{\boldsymbol{\phi}^*}{\|\boldsymbol{\phi}\|} = \sqrt{\rho\tau_{\rm p}}\vect{h}+\vect{n}^{\rm pilot},
\end{align}
where $\vect{n}^{\rm pilot}=\vect{N}^{\rm pilot}\boldsymbol{\phi}^*/\|\boldsymbol{\phi}\|\sim \CN(\vect{0},\vect{I}_M)$.
In this paper, we will use the RS-LS channel estimation framework\cite{demir2022channel}, which outperforms the conventional LS estimator significantly without any knowledge of the spatial correlation matrix. We will derive an achievable lower bound on the average SE, which only depends on the statistics of the $\vect{R}$-matrices among the population of possible UEs. Later, we will optimize the pilot length $\tau_p$ to maximize the average SE for a given UE channel gain $\beta$.

We let $1\leq r \leq M$ denote the rank of $\vect{R}_{\rm iso}$, i.e., $\rank\left(\vect{R}_{\rm iso}\right)=r$. The \emph{compact} eigendecomposition is denoted as $\vect{R}_{\rm iso}=\vect{U}_1\vect{\Lambda}_1\vect{U}_1^{\Htran}$, where the diagonal matrix $\vect{\Lambda}_1 \in \mathbb{C}^{r \times r}$ contains the non-zero eigenvalues and the columns of $\vect{U}_1\in \mathbb{C}^{M \times r}$ contains the corresponding orthonormal eigenvectors.\footnote{Using the compact eigendecomposition is the essence of obtaining reduced-subspace representation and corresponding RS-LS channel estimator.} 
It was proved in \cite{demir2022channel} that any spatial correlation matrix of the form in \eqref{eq:spatial-correlation2} is spanned by the columns of $\vect{U}_1$. Hence, any UE's channel vector can be expressed as $\vect{h}=\vect{U}_1\vect{x}$ for a zero-mean complex Gaussian vector $\vect{x}\in \mathbb{C}^{r}$. The reduced-subspace $\vect{U}_1$ is common for any potential UE, however, the covariance matrix of $\vect{x}$ is UE-dependent and unknown. The so-called \emph{RS-LS} estimate of $\vect{h}$ is
\begin{align} \label{eq:RS-LS-estimate-approx}
    \widehat{\vect{h}}= \frac{\vect{U}_1\vect{U}_1^{\Htran}\vect{y}^{\rm pilot}}{\sqrt{\rho\tau_{\rm p}}} = { \underbrace{\vect{U}_1\vect{U}_1^{\Htran}\vect{U}_1\vect{x}}_{=\vect{h}}} + \underbrace{\frac{\vect{U}_1\vect{U}_1^{\Htran}\vect{n}^{\rm pilot}}{\sqrt{\rho\tau_{\rm p}}}}_{\triangleq \vect{w}},
\end{align}
where the channel estimation error is $\vect{w} \sim \CN(\vect{0},\vect{C})$ with
\begin{align} \label{eq:estimation-error-C}
\vect{C} &= \mathbb{E}\{\vect{w}\vect{w}^{\Htran} \} = \frac{\vect{U}_1\vect{U}_1^{\Htran}\mathbb{E}\left\{\vect{n}^{\rm pilot}\left(\vect{n}^{\rm pilot}\right)^{\Htran}\right\}\vect{U}_1\vect{U}_1^{\Htran}}{\rho\tau_{\rm p}} \nonumber\\
&= \frac{\vect{U}_1\vect{U}_1^{\Htran}}{\rho \tau_{\rm p}}
\end{align}
and $\tr\left(\vect{C}\right)=r/(\rho\tau_{\rm p})$. The channel estimate and estimation error are not uncorrelated since the RS-LS estimator is not the MMSE estimator. This will make the SE analysis different than when considering the conventional MMSE estimator.

 \section{Average Uplink Spectral Efficiency and Pilot Length Optimization}
 
 The received signal during uplink data transmission is  
\begin{align}
    \vect{y} = \sqrt{\rho}\vect{h}s+\vect{n},
\end{align}
where $s\in\mathbb{C}$ is the information signal of the UE with $\mathbb{E}\{|s|^2\}=1$ and $\vect{n}\sim \CN(\vect{0},\vect{I}_M)$ is the independent receiver noise. The BS applies maximum-ratio combining, thus it multiplies the received signal with the estimate $\widehat{\vect{h}}^{\Htran}$ to obtain the effective single-input single-output (SISO) channel
\begin{align}
y = \widehat{\vect{h}}^{\Htran}\vect{y} = \sqrt{\rho}\left(\vect{h}+\vect{w}\right)^{\Htran}\vect{h}s+\left(\vect{h}+\vect{w}\right)^{\Htran}\vect{n}.
\end{align}
The average channel gain $\mathbb{E}\{\vect{h}^{\Htran}\vect{h}\}=\beta\tr(\vect{R})=M\beta$ among all antennas is assumed known at the BS since it is easy to estimate such a scalar. Here, $\beta>0$ denotes the average channel gain. Hence, we can write the above SISO channel as
 \begin{align}
    y= & \sqrt{\rho}\mathbb{E}\{\vect{h}^{\Htran}\vect{h}\}s + \sqrt{\rho}\left(\vect{h}^{\Htran}\vect{h}-  \mathbb{E}\{\vect{h}^{\Htran}\vect{h}\}\right)s \nonumber\\
&+\sqrt{\rho}\vect{w}^{\Htran}\vect{h}s+\left(\vect{h}+\vect{w}\right)^{\Htran}\vect{n} \nonumber \\
=& \underbrace{\sqrt{\rho}M\beta}_{\triangleq h} s+ \underbrace{\sqrt{\rho}(\vect{h}^{\Htran}\vect{h}+\vect{w}^{\Htran}\vect{h}-M\beta)s+\left(\vect{h}+\vect{w}\right)^{\Htran}\vect{n}}_{\triangleq \upsilon},
 \end{align}
 where $h$ is the known channel at the BS and $\upsilon$ is the effective noise. This channel representation is equivalent to a discrete memoryless channel with the input $s$,  output $y$, and uncorrelated noise. Utilizing \cite[Cor.~1.3]{massivemimobook}, the following lemma provides an achievable SE for this channel.

\begin{lemma} \label{lemma:infinite} 
An achievable SE of a particular UE with the channel $\vect{h}\sim \CN(\vect{0},\beta\vect{R})$ with RS-LS channel estimation is
\begin{align} \label{eq:SE}
   & \mathrm{SE} = \frac{\tau_c-\tau_p}{\tau_c}\times \nonumber \\&\log_2\left(1+\frac{\rho M^2 \beta^2 }{\rho\beta^2\tr\left(\vect{R}^2\right) +M\beta + \frac{1}{\tau_p} \left( M\beta+\frac{r}{\rho}\right)}\right) .
\end{align}
\end{lemma}
\vspace{-2mm}
\begin{proof}
    The proof is provided in Appendix~\ref{appendix1}.
\end{proof}

 \begin{figure*}[h]
\begin{align} \label{eq:average1}
   \mathrm{SE}_{\m avg}= \mathbb{E}_{\vect{R}}\{\mathrm{SE}\} &=\frac{\tau_c-\tau_p}{\tau_c} 
    \mathbb{E}_{\vect{R}}\left\{ \log_2\left(1+\frac{\rho M^2 \beta^2 }{\rho\beta^2\tr\left(\vect{R}^2\right) +M\beta + \frac{1}{\tau_p} \left( M\beta+\frac{r}{\rho}\right)}\right)\right\} \\ \label{eq:average2}
    &\geq \frac{\tau_c-\tau_p}{\tau_c} 
   \log_2\left(1+\frac{\rho M^2  }{\rho\mathbb{E}\left\{\tr\left(\vect{R}^2\right)\right\} +\frac{M}{\beta} + \frac{1}{\tau_p}  \left(\frac{M}{\beta}+\frac{r}{\rho\beta^2}\right)}\right)
\end{align}
\hrulefill
\vspace{-6mm}
\end{figure*}

Even if the RS-LS channel estimator does not require knowledge of $\vect{R}$, the SE expression in Lemma~\ref{lemma:infinite} depends on it through the scalar $\tr(\vect{R}^2)$.
Now, we will derive a lower bound on the average SE that only depends on $\beta$ and the statistics of the normalized spatial correlation matrix for the entire population of UEs with the same $\beta$.

The average SE with RS-LS channel estimation is computed in \eqref{eq:average1} at the top of the next page. We have used the subscript $\vect{R}$ to emphasize that the expectation is taken over all possible normalized spatial correlation matrices $\vect{R}$ in the user population. Note that the channel gain $\beta$ and all the other parameters are treated as fixed, except for $\tau_p$, which we will optimize in the following part. To make the optimization manageable, the lower bound in  \eqref{eq:average2} on the average SE is obtained by applying Jensen's inequality to the convex function $\log_2(1+\frac{a}{b+cx})$ of $x$ for some positive constants $a$, $b$, and $c$, respectively.\footnote{It holds that $\mathbb{E}\{\log_2(1+\frac{a}{b+cx})\}\geq \log_2(1+\frac{a}{b+c\mathbb{E}\{x\}} )$ from Jensen's inequality.} Now, the average SE can be written as a function of $\tau_p$ as
\begin{align} \label{eq:SE-avg-taup}
    \mathrm{SE}_{\rm avg}(\tau_p) = \frac{\tau_c-\tau_p}{\tau_c}\log_2\left(1+\frac{1}{\mathsf{A}+\frac{\mathsf{B}}{\tau_p}}\right),
\end{align}
where the constants $\mathsf{A}$ and $\mathsf{B}$ that are independent of $\tau_p$ are
\begin{align}
&\mathsf{A} = \frac{1}{M^2}\mathbb{E}\left\{\tr\left(\vect{R}^2\right)\right\} +\frac{1}{\rho M\beta}, \\
&\mathsf{B} = \frac{1}{\rho M^2} \left(\frac{M}{\beta}+\frac{r}{\rho\beta^2}\right).
\end{align}

 \begin{lemma}\label{lemma:concave}
 $\mathrm{SE}_{\rm avg}(\tau_p)$ in \eqref{eq:SE-avg-taup} is a strictly concave function of $\tau_p$, hence, it has a unique global maximum $\tau_p^{\rm opt}$. Moreover, it holds that $0<\tau_p^{\rm opt}<\tau_c$.
 \end{lemma}
 \begin{proof}
    The proof is provided in the Appendix~\ref{appendix2}.
\end{proof}

When the SNR is so high that $\mathsf{B}$ has a much smaller value compared to $\mathsf{A}$ in \eqref{eq:SE-avg-taup}, then the optimal $\tau_p$ is expected to be a small value, possibly $\tau_p\approx 1$. On the other hand, when the SNR is small, we expect that a larger number of pilots is needed to maximize the average SE. In the low-SNR regime, we can approximate the average SE in \eqref{eq:SE-avg-taup} as
\begin{align}
      \mathrm{SE}_{\rm avg}(\tau_p)  \approx  \log_2(e)\frac{\tau_c-\tau_p}{\tau_c}\frac{1}{\mathsf{A}+\frac{\mathsf{B}}{\tau_p}}, \label{eq:approximation}
\end{align}
where we used $\log_2(1+x)\approx \log_2(e)x$ for $x\approx 0$. Taking the derivative of the above expression with respect to $\tau_p$ and equating it to zero, we can compute the optimal pilot length in the low-SNR regime as
\begin{align}
    \tau_p^{\star} \approx \frac{\sqrt{\mathsf{B}\left(\mathsf{B}+\mathsf{A}\tau_c\right)}-\mathsf{B}}{\mathsf{A}}. \label{eq:approximate-solution}
\end{align}
The optimal integer-valued pilot length is one of the two closest integers to $\tau_p^{\star} $ that gives the largest average SE.

\section{Numerical Results}

In this section, we will numerically analyze the impact of the number of antennas and channel gain on the average SE. Moreover, the SE achieved with the optimal pilot lengths obtained by the exact expression in \eqref{eq:average1}, the lower bound in \eqref{eq:average2}, and the approximate solution in \eqref{eq:approximate-solution} will be compared.

The local scattering model presented in closed-form in \cite[Lem.~1]{demir2022channel} is utilized to model the spatial correlation. There are $N=5$ clusters with uniformly randomly distributed powers in $[0,1]$ that are normalized by the sum of all cluster powers. For each cluster, the nominal azimuth and elevation angles are both generated randomly in $[-\pi/3,\pi/3]$ following a uniform distribution. The per-cluster angular standard deviations are $\sigma_{\varphi} = \sigma_{\theta} = 5^{\circ}$ and each antenna in the array is a directive antenna with the cosine pattern $\cos(\varphi)\cos(\theta)$. The pilot transmit power is $100$\,mW and the noise variance is $-94$\,dBm corresponding to $20$\,MHz bandwidth and a noise figure of $7$\,dB. The number of channel uses in a coherence block is $\tau_c=200$.

In Fig.~\ref{fig:fig2}, we plot the exact average UE expression in \eqref{eq:average1} and the lower bound in \eqref{eq:average2} in terms of the pilot length $\tau_p$. The channel gain $\beta$ is selected to have the SNR $\beta\rho=-20$\,dB. The results are averaged over many realizations of the channel statistics. We compare two antenna arrays with the same aperture: i) $M_{\rm H}=M_{\rm V}=12$ with $\Delta=\lambda/4$ and ii) $M_{\rm H}=M_{\rm V}=24$ with $\Delta=\lambda/8$. The approximate solution in \eqref{eq:approximate-solution}, which is rounded to the closest integer that gives the highest average SE, is demonstrated by a star for each case. As the figure shows, the lower bound matches well with the exact expression. Moreover, there is only one $\tau_p$ that maximizes the average SE, which verifies Lemma~\ref{lemma:concave}. Another important observation is that the optimal $\tau_p$ is smaller when having more densely deployed antennas, which corresponds to increased spatial correlation. Consequently, this results in improved noise rejection capabilities when employing the RS-LS channel estimator due to the reduced dimension of the subspace that all possible channel realizations span.

In Fig.~\ref{fig:fig3}, we plot the cumulative distribution function (CDF) of the SE when $M_{\rm H}=M_{\rm V}=24$ and $\Delta=\lambda/8$ (as in the previous figure) using different pilot lengths. In addition to the optimal $\tau_p$ values that are obtained by maximizing the exact expression, lower bound, and low-SNR approximation of the average SE, $\tau_p=10$ is also shown as a reference, which represents an arbitrary value of $\tau_p$. Optimizing $\tau_p$ using the exact expression or the lower bound provides almost always a higher SE. On the other hand, there is a small SE reduction when using the low-SNR approximation. Much smaller SE values are obtained when using an arbitrary pilot length, which is $\tau_p=10$ in this figure. 

In Figs.~\ref{fig:fig4} and ~\ref{fig:fig5}, we increase the SNR to $-10$\,dB and repeat the experiment in Figs.~\ref{fig:fig2} and ~\ref{fig:fig3}. As expected, higher SE values are obtained due to an increase in the SNR, and fewer pilot symbols are required to maximize the average SE. Hence, the reference pilot length in Fig.~\ref{fig:fig5} is selected as $\tau_p=1$. This time, as shown in Fig.~\ref{fig:fig4}, a slight gap between the lower bound and the exact expression is observed. However, it is negligible, and optimizing $\tau_p$ based on the lower bound and exact expression almost give the same CDF of the SE as in Fig.~\ref{fig:fig5}. Moreover, using the low-SNR approximation provides very close performance. 

 \begin{figure}[t!] 
\centering
	\includegraphics[trim={0cm 0.cm 0cm 1cm},clip,width=3.4in]{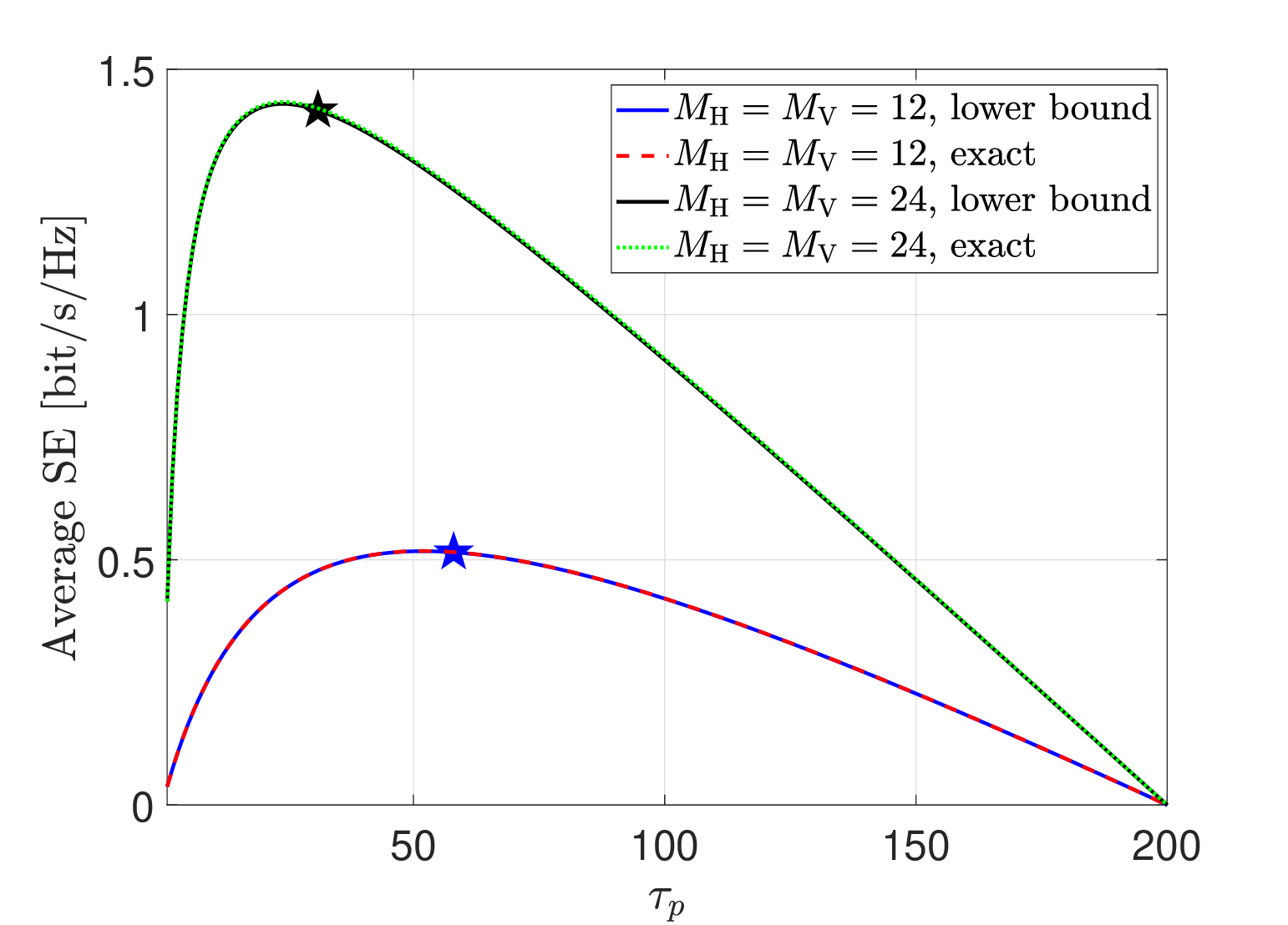}
			\vspace{-0.4cm}
			\caption{The exact expression and proposed lower bound on the average SE versus the pilot length for different numbers of antennas and $-20$\,dB SNR. The inter-antenna separation is $\Delta=\lambda/4$ when $M_{\rm H}=M_{\rm V}=12$ and $\Delta=\lambda/8$ when  $M_{\rm H}=M_{\rm V}=24$, respectively.} \label{fig:fig2} \vspace{-4mm}
\end{figure}

 \begin{figure}[t!] 
\centering
	\includegraphics[trim={0cm 0.cm 0cm 1cm},clip,width=3.4in]{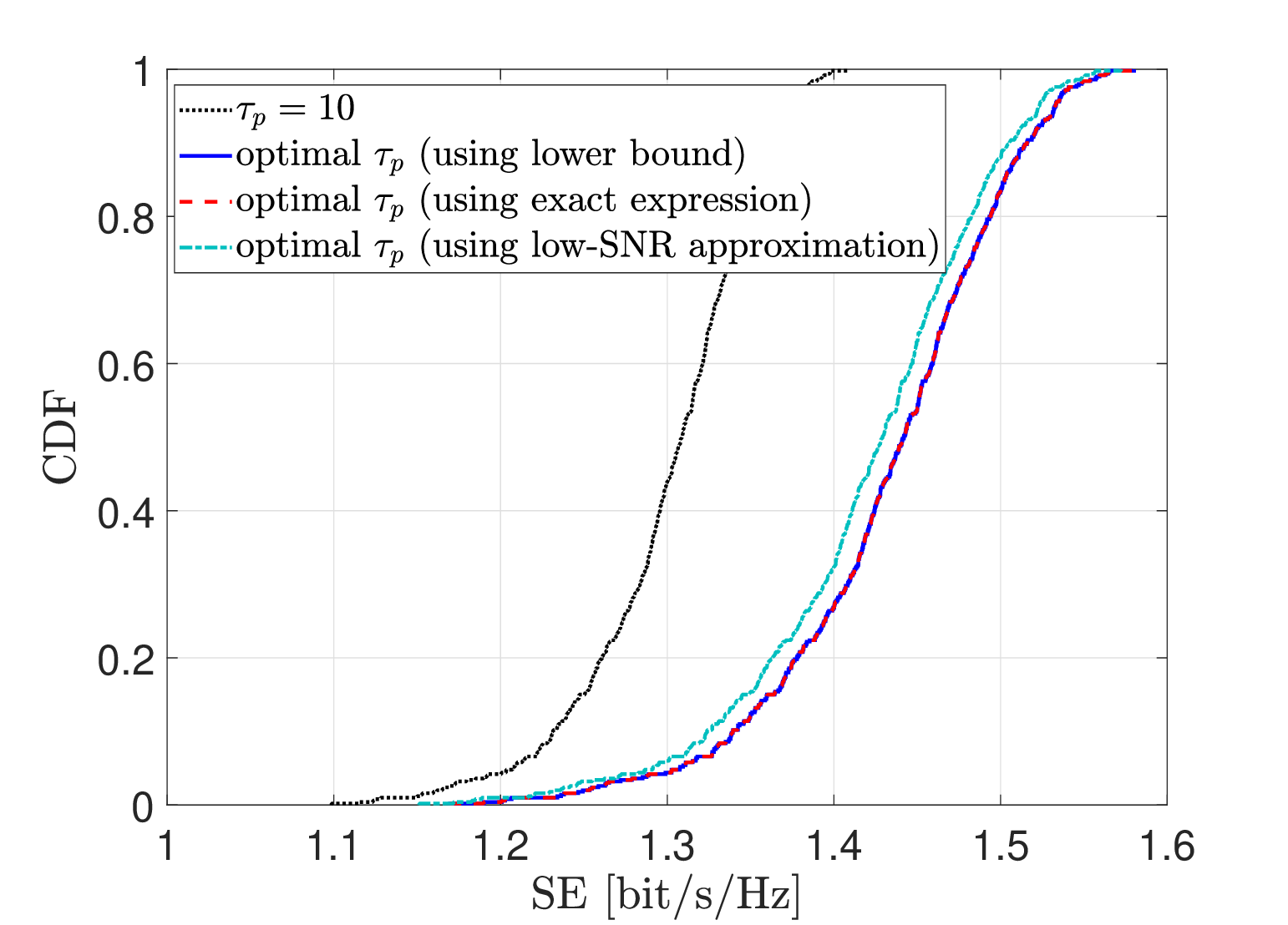}
			\vspace{-0.4cm}
			\caption{CDFs of the SE when $M_{\rm H}=M_{\rm V}=24$, $\Delta=\lambda/8$, and $-20$\,dB SNR using different pilot lengths. } \label{fig:fig3} \vspace{-4mm}
\end{figure}

 \begin{figure}[t!] 
\centering
	\includegraphics[trim={0cm 0.cm 0cm 1cm},clip,width=3.4in]{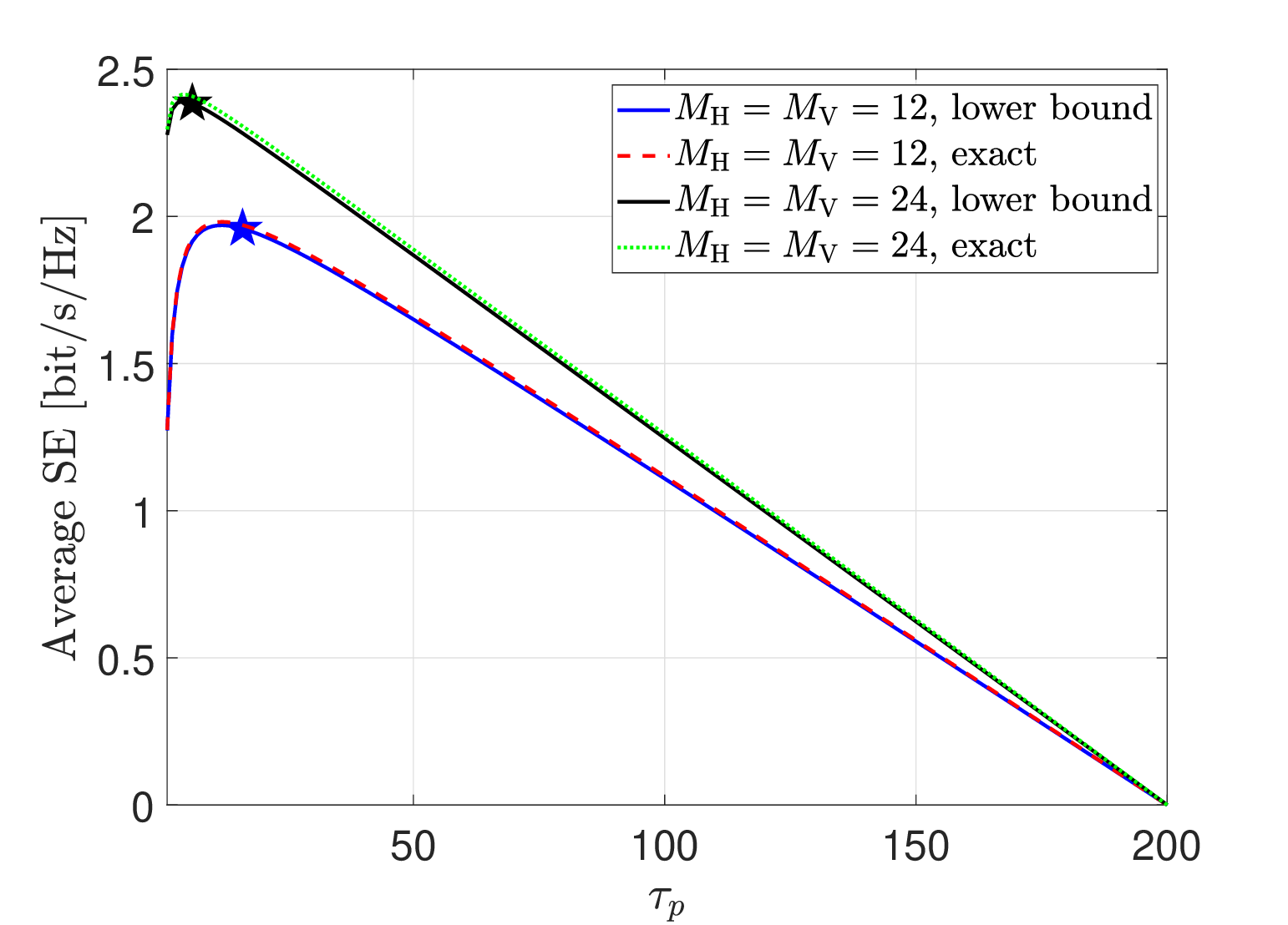}
			\vspace{-0.4cm}
			\caption{Exact expression and proposed lower bound to the average SE versus the pilot length for different numbers of antennas and $-10$\,dB SNR. The inter-antenna separation is $\Delta=\lambda/4$ when $M_{\rm H}=M_{\rm V}=12$ and $\Delta=\lambda/8$ when  $M_{\rm H}=M_{\rm V}=24$, respectively.} \label{fig:fig4} \vspace{-4mm}
\end{figure}

 \begin{figure}[t!] 
\centering
	\includegraphics[trim={0cm 0.cm 0cm 1cm},clip,width=3.4in]{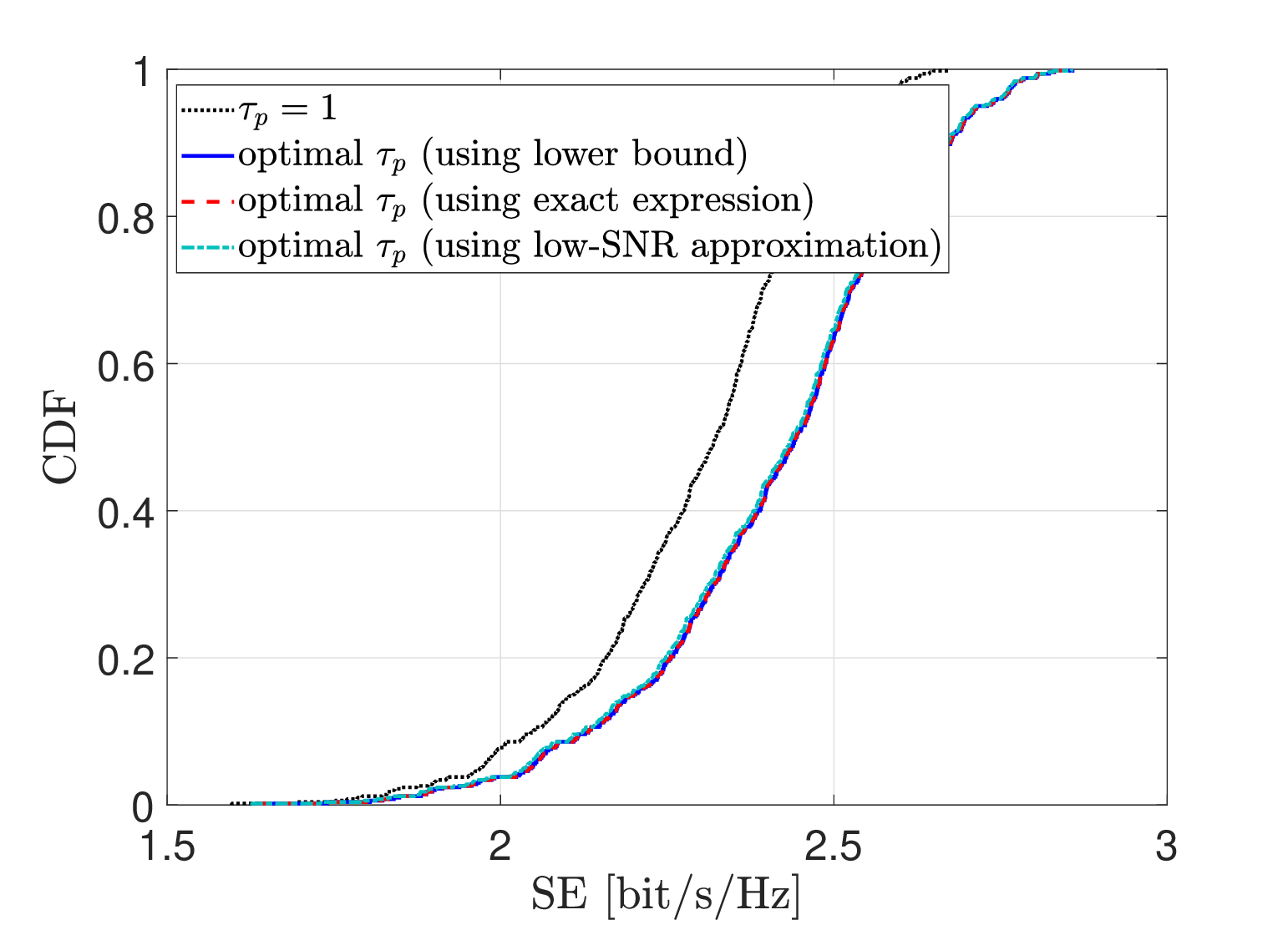}
			\vspace{-0.4cm}
			\caption{CDFs of the SE when $M_{\rm H}=M_{\rm V}=24$, $\Delta=\lambda/8$, and $-10$\,dB SNR using different pilot lengths.} \label{fig:fig5} \vspace{-4mm}
\end{figure}

 \begin{figure*}[t]
\setcounter{equation}{22}

\begin{align} \label{eq:long}
\mathbb{E}\left\{|\upsilon|^2\right\}\stackrel{(a)}=&\rho \mathbb{E}\left\{ \left| \vect{h}^{\Htran}\vect{h} + \vect{w}^{\Htran}\vect{h}-M\beta\right|^2 \right\} \mathbb{E}\{|s|^2\} +\mathbb{E}\left\{  \left(\vect{h}+\vect{w}\right) ^{\Htran}\vect{n}\vect{n}^{\Htran}\left(\vect{h}+\vect{w}\right)  \right\} \nonumber \\
\stackrel{(b)}=&\rho\mathbb{E}\left\{\vect{h}^{\Htran}\vect{h}\vect{h}^{\Htran}\vect{h}\right\}+\rho\mathbb{E}\left\{\vect{w}^{\Htran}\vect{h}\vect{h}^{\Htran}\vect{w}\right\}+\rho M^2\beta^2 +2\rho\Re\left(\underbrace{\mathbb{E}\left\{\vect{w}^{\Htran}\vect{h}\vect{h}^{\Htran}\vect{h}\right\}}_{=\mathbb{E}\left\{\vect{w}^{\Htran}\right\}\mathbb{E}\left\{\vect{h}\vect{h}^{\Htran}\vect{h}\right\}=0}\right)-2\rho M\beta\underbrace{\mathbb{E}\left\{\vect{h}^{\Htran}\vect{h}\right\}}_{=M\beta} \nonumber\\
&-2\rho M\beta\Re\left(\underbrace{\mathbb{E}\left\{\vect{w}^{\Htran}\vect{h}\right\}}_{=\mathbb{E}\left\{\vect{w}^{\Htran}\right\}\mathbb{E}\left\{\vect{h} \right\}=0}\right) +\mathbb{E}\left\{ \left(\vect{h}+\vect{w}\right)^{\Htran}\underbrace{\mathbb{E}\left\{\vect{n}\vect{n}^{\Htran}\right\}}_{=\vect{I}_M}\left(\vect{h}+\vect{w}\right)\right\} \nonumber \\
\stackrel{(c)}=& \rho\mathbb{E}\left\{\vect{h}^{\Htran}\vect{h}\vect{h}^{\Htran}\vect{h}\right\}+\rho\mathbb{E}\left\{  \vect{w}^{\Htran}\underbrace{\mathbb{E}\left\{\vect{h}\vect{h}^{\Htran}\right\}}_{=\beta\vect{R}}\vect{w}\right\}   -\rho M^2 \beta^2 + \underbrace{\mathbb{E}\left\{\vect{h}^{\Htran}\vect{h}\right\}}_{=M\beta}+ \mathbb{E}\left\{\vect{w}^{\Htran}\vect{w}\right\} \nonumber \\
\stackrel{(d)}=&\underbrace{\rho\beta^2\left(\tr\left(\vect{R}\right)\right)^2}_{=\rho M^2 \beta^2}+\rho\beta^2\tr\left(\vect{R}^2\right)+\rho\beta\tr\left(\mathbb{E}\left\{\vect{w}\vect{w}^{\Htran}\right\}\vect{R}\right)-\rho M^2\beta^2+M\beta+\tr\left(\mathbb{E}\left\{\vect{w}\vect{w}^{\Htran}\right\}\right) \nonumber\\
\stackrel{(e)}=&\rho\beta^2\tr\left(\vect{R}^2\right) +\rho\frac{\beta\tr\left(\vect{U}_1\vect{U}_1^{\Htran}\vect{R}\right)}{\rho \tau_{\rm p}}+M\beta+\frac{\tr\left(\vect{U}_1\vect{U}_1^{\Htran}\right)}{\rho\tau_p} \stackrel{(f)}= \rho\beta^2\tr\left(\vect{R}^2\right) +M\beta + \frac{1}{\tau_p} \left(M \beta+\frac{r}{\rho}\right)
\end{align}
\hrulefill
\vspace{-6mm}
 \end{figure*}

\section{Conclusions}

\vspace{-1mm}

In this paper, we have derived a novel SE expression applicable when using the practical RS-LS channel estimator for communication between an ELAA and a single-antenna UE. This estimator improves estimation quality without requiring UE-specific channel statistics.
Furthermore, we introduced an expression for the average SE, considering a certain UE distribution and, consequently, a distribution of the normalized spatial correlation matrices. To facilitate fast optimization of the pilot length, we also derived a lower bound on the average SE. We showed that there exists a unique pilot length that maximizes this lower bound. As the SNR decreases, the demand for additional pilot resources becomes apparent, leading us to derive a closed-form expression for an approximate pilot length that maximizes the SE under low-SNR conditions.

In scenarios with densely deployed antennas, the improved noise rejection capability of the RS-LS channel estimator reduces the requirement for pilot symbols. Our derived lower bound closely approximates optimal performance, and optimizing the pilot length based on it typically yields the best results. The optimal pilot length obtained from the low-SNR approximation offers near-optimal performance in most cases.

\vspace{-1mm}
\appendices

\section{Proof of Lemma~\ref{lemma:infinite}} \label{appendix1}
\vspace{-1mm}

The considered channel is a discrete memoryless channel with the input $s$ and output $y$ and uncorrelated noise, as in \cite[Cor.~1.3]{massivemimobook}. In addition, zero-mean $\upsilon$ and $s$ are uncorrelated as shown below:
\setcounter{equation}{20}
\begin{align}
\mathbb{E}\{s^*\upsilon\} &=\mathbb{E}\left\{|s|^2\sqrt{\rho}\left(\vect{h}^{\Htran}\vect{h}+\vect{w}^{\Htran}\vect{h}-M\beta\right)\right\}\nonumber\\
&\quad+\mathbb{E}\left\{s^*\left(\vect{h}+\vect{w}\right)^{\Htran}\vect{n}\right\}\nonumber\\
&=\mathbb{E}\left\{|s|^2\right\}\underbrace{\sqrt{\rho}\left(\mathbb{E}\left\{\vect{h}^{\Htran}\vect{h}\right\}+\mathbb{E}\left\{\vect{w}^{\Htran}\right\}\mathbb{E}\left\{\vect{h}\right\}-M\beta\right)}_{=0}\nonumber\\
&\quad+\underbrace{\mathbb{E}\{s^*\}}_{=0}\mathbb{E}\left\{\left(\vect{h}+\vect{w}\right)^{\Htran}\vect{n}\right\},
\end{align}
where we used that $s$, $\vect{h}$, $\vect{w}$, and $\vect{n}$ are mutually independent. It follows from \cite[Cor.~1.3]{massivemimobook} that an achievable SE is

\begin{align} \label{eq:SE-appendix}
    \mathrm{SE} = \frac{\tau_c-\tau_p}{\tau_c}\log_2\left(1+\frac{|h|^2}{\mathbb{E}\{|\upsilon|^2\}}\right) ,
\end{align}
where we have also included the pre-log factor $\frac{\tau_c-\tau_p}{\tau_c}$ to account for the fact that only $\tau_c-\tau_p$ symbols are used for data transmission in each coherence block. The denominator term is computed as shown in \eqref{eq:long} at the top of the next page.

In $(a)$ in \eqref{eq:long}, we used the independence of the zero-mean random variable $s$ from $\vect{h}$, $\vect{w}$, and $\vect{n}$. In $(b)$ and $(c)$, we used the independence of the zero-mean random vectors $\vect{w}$, $\vect{h}$, and $\vect{n}$, respectively. In $(d)$, the result $\mathbb{E}\left\{\vect{h}^{\Htran}\vect{h}\vect{h}^{\Htran}\vect{h}\right\}=\beta^2\left(\tr\left(\vect{R}\right)\right)^2+\beta^2\tr\left(\vect{R}^2\right)$ from \cite[Lem.~B.14]{massivemimobook} and the cyclic shift property of the trace are utilized. Finally, in $(e)$ and $(f)$, $\vect{C}$ in \eqref{eq:estimation-error-C} and $\tr(\vect{C})=r/(\rho\tau_{\rm p})$ are inserted, respectively. In $(f)$, it is also noted that $\tr\left(\vect{U}_1 \vect{U}_1^{\Htran}\vect{R}\right)=\tr\left(\vect{R}\right)=M$ since $\vect{R}$ is in the span of $\vect{U}_1$.       

Substituting $|h|^2=\rho M^2\beta^2$ and $\mathbb{E}\{|\upsilon|^2\}=\rho\beta^2\tr\left(\vect{R}^2\right) +M\beta + \frac{1}{\tau_p} (M\beta +\frac{r}{\rho})$ into \eqref{eq:SE-appendix}, the SE expression in \eqref{eq:SE} is obtained.

\vspace{-2mm}
\section{Proof of Lemma~\ref{lemma:concave}} \label{appendix2}

The proof follows from the second-order derivative of $\mathrm{SE}_{\rm avg}(\tau_p)$ in \eqref{eq:SE-avg-taup}, which is 
\setcounter{equation}{23}
    \begin{align}
    -\dfrac{\mathsf{B}\left(\left(\left(2\mathsf{A}^2+2\mathsf{A}\right)\tau_c+\left(2\mathsf{A}+1\right)\mathsf{B}\right)\tau_p+\left(2\mathsf{A}+1\right)\mathsf{B}\tau_c+2\mathsf{B}^2\right)}{\ln\left(2\right)\,\tau_c\left(\mathsf{A}\tau_p+\mathsf{B}\right)^2\left(\left(\mathsf{A}+1\right)\tau_p+\mathsf{B}\right)^2}
    \end{align}
    and is negative for any given $\mathsf{A}>0$ and $\mathsf{B}>0$. To see where the optimal solution lies, we notice that $\mathrm{SE}_{\rm avg}=0$ when $\tau_p=0$ and $\tau_p=\tau_c$. Since the average SE is a strictly concave and non-negative function, it must attain its maximum at some $0<\tau_p^{\rm opt}<\tau_c$.  

\bibliographystyle{IEEEtran}
\bibliography{IEEEabrv,refs}

\end{document}